\documentclass[12pt,onecolumn]{IEEEtran}
\usepackage{cite}
\usepackage{amsmath,amssymb,amsfonts,amsthm,mathtools}
\usepackage{algorithm,algorithmic}
\usepackage{pgfplots}
\usepackage{textcomp}
\usepackage{soul}
\pgfplotsset{compat=newest}
\newtheorem{theorem}{Theorem}

\newcommand{\vc}[1]{\mathbf{#1}}

\newcommand{\hb}{\vc{h}}
\newcommand{\gb}{\vc{g}}
\newcommand{\ab}{\vc{a}}

\newcommand{\Thetab}{\vc{\Theta}}
\newcommand{\wb}{\vc{w}}

\newcommand{\cb}{\vc{c}}
\newcommand{\xb}{\vc{x}}

\newcommand{\Sb}{\vc{S}}
\newcommand{\Xb}{\vc{X}}
\newcommand{\Tb}{\vc{T}}

\newcommand{\Fb}{\vc{F}}

\newcommand{\Wb}{\vc{W}}

\begin{document}

\title{\ \\ Secrecy Rate Maximization in Multi-IRS mmWave Networks \\}

\author{  Anahid Rafieifar, and S. Mohammad Razavizadeh
\\ 
School of Elec. Eng., Iran University of Science and Technology (IUST)\\
anahid{\_}rafieifar@elec.iust.ac.ir, smrazavi@iust.ac.ir}
\maketitle
\begin{abstract}
This paper investigates the problem of increasing the security at the physical layer of a millimeter Wave (mmWave) network equipped with several Intelligent Reflecting Surfaces (IRSs). In this network, multiple IRSs help the Base Station (BS) to reach the signal to the desired user and at the same time maintain the security of the network i.e. securing the signal from receiving by the unallowable eavesdropper. The target of the proposed scheme is to maximize the secrecy rate by jointly optimizing the active beamforming at the BS and passive beamforming at the IRSs. This leads to a non-convex optimization problem which we solve by decomposing into two sub-problems. The sub-problems alternatively solve the active and passive beamforming design problems using the Semi-Definite Relaxation (SDR) technique. Finally, simulations  are done to assess the performance of the proposed algorithm. These results show the superiority of using multiple IRSs in the enhancement of the secrecy rate in the wireless networks that operate in the mmWave frequency bands.
\end{abstract}

\begin{IEEEkeywords}
	Intelligent Reflecting Surface (IRS), Reconfigurable Intelligent surfaces (RIS), multiple IRSs, multi-IRS, multi-RIS, millimeter wave (mmWave) communications, physical layer security, eavesdropping.
\end{IEEEkeywords}
\newpage
\section{Introduction}
One of the key technologies in current and future generations of wireless networks is millimeter Wave (mmWave) communications \cite{hs2}. 
The main advantages of the mmWave frequency bands are higher bandwidth, higher frequency reuse, less interference, and smaller antenna array dimensions. 
Despite these benefits, mmWave communications are more sensitive to blockage and suffer from a large attenuation compared with lower frequencies \cite{hs2}.
On the other hand, Intelligent Reflecting Surface (IRS) or Reconfigurable Intelligent Surface (RIS) is a novel technology that has been recently proposed to improve the performance of wireless communication networks. \cite{hs3}. 
These are metasurfaces that can control the propagation direction of the wireless signals to focus them on a point of interest and improve the wireless link qualities \cite{emil1}. 
The IRSs can also be employed in the mmWave networks to compensate for the blockage effect and high path loss in these networks.  
\subsection{Related Work}
Passive beamforming design for the IRSs in the microwave bands is extensively studied in literature \cite{passive2,hs5}. 
Wang et al. in \cite{hs5}, propose an algorithm to maximize users' received power by joint active and passive beamforming in multi-IRSs assisted communications. 
In \cite{hs4} it is shown that an increase in the number of reflecting elements of an IRS leads to the BS’s power can be decreased quadratically.
Also in \cite{hs6}, an optimized beamforming and power allocation method is proposed for maximizing the total rate in the mmWave systems where multi-IRSs are employed. In \cite{conferance1}, the authors focused on minimizing BS's transmitted power and jointly optimizing active and passive beamforming.

\textcolor{black}{Due to the passive nature and a high number of reflecting elements in IRSs, channel estimation in the IRS-assisted systems becomes a challenging issue that has been investigated in some recent literature \cite{ref4, ref5, ref3}. In most of this literature, channels are partially estimated, and using the deep learning approach are extrapolated. In \cite{ref4}, the authors considered an IRS that is equipped with some active elements and is capable to acquire their partially Channel State Information (CSI). The locations of these active elements are derived based on probabilistic sampling theory. Two deep learning-based schemes, namely, channel extrapolation and beam searching are also designed. The authors in \cite{ref5} proposed an Ordinary Differential Equation based Convolutional Neural Network (ODE-CNN) approach to acquire full CSI from extrapolation of the partial CSI. But in contrast to \cite{ref4}, they only use passive elements and turn on the fractional part of IRSs’ elements. They also show that their proposed method can have a faster convergence rate and achieve a better solution than cascaded CNN. In \cite{ref3}, to reduce IRS channel training overhead in a system consists of a large number of BS antennas or IRS reflecting elements, three types of channel extrapolation namely, antenna, frequency, and physical terminal extrapolations are proposed.
} 

The problem of physical layer security in IRS-aided networks is also studied in recent literature. For example, joint optimizing of the BS beamforming and IRS phase shifts for secrecy rate maximization is considered in \cite{hs7,hs8,hs9}. Furthermore, using artificial noise for enhancing the secrecy rate of an IRS network is studied in \cite{hs10}. By employing a physical layer security approach, the authors in \cite{hs11} have jointly designed beamforming matrix and artificial noise signals at the base station, and the phase shifts of the IRS's to improve the sum rate of a multi-IRS system. 
An IRS-assisted mmWave network is studied in \cite{xiu2020secrecy} in which the network secrecy rate is enhanced by applying an alternating optimization-based method and joint optimization of the IRS phase shift and the beamforming at the transmitter. To do this end, the authors employed two approaches namely, element-wise block coordinate descent and successive convex approximation (SCA). The combination of physical layer security and channel estimation in an IRS-assisted system is also investigated \cite{ref1},\cite{ref2}.
 \textcolor{black}{The authors in \cite{ref1} perform secure beamforming by employing a deep reinforcement learning approach under the time-varying channel assumption. To improve the learning efficiency they use post-decision state and prioritized experience replay. Their results show a considerable improvement in secrecy rate and QoS satisfaction probability compared with traditional methods. Yang et al. considered an IRS-aided anti-jamming system, where IRS can mitigate the jamming interference by jointly power allocation at the BS and phase shift design at the IRS using the fuzzy win or learn fast-policy hill-climbing approach. Simulation results demonstrated the IRS-aided system enhances the system rate and transmission protection level compared with existing approaches \cite{ref2}.}

However, the problem of multi-IRS networks operating in the mmWave frequency bands is still an
open problem and on the other hand, the problem of secrecy rate maximization in multi-IRS mmWave has not been widely studied in the literature.

\subsection{Contribution}
%an provide multiple paths for the received signal,
In this paper, motivated by the following facts, we study the problem of physical layer security in a multi-IRS aided millimeter-wave network.
In this mmWave network, the IRSs are utilized to create multiple links between a multiple-antenna BS and a user, and at the same time, an eavesdropper listens to these signals. Both the user and eavesdropper utilize a single antenna for signal transmission and reception. 
The target of the proposed method is maximizing the network secrecy rate by jointly optimizing the beamforming coefficients at the base station and phase shifts at the IRSs. 
This results in a non-convex optimization problem and to solve it we design an efficient method that is based on alternating method and Semi-Definite Relaxation (SDR) technique.
Finally, we present simulation results to show the advantages of using multiple IRSs in the mmWave network for secrecy rate enhancement. 

Our main contributions can be  summarized as follows:
\begin{itemize}
	\item In most papers, a single IRS in a wireless network is considered, But we study a network equipped with multiple IRS. By increasing the number of IRSs in the mmWave band, the extra paths from the BS to the user are created, which improves the received power in the user and thus increases the secrecy rate. Due to the high path loss and blockage sensitivity in the mmWave frequency bands, it is more necessary to create more paths in this band compared to the microwave band.
	
\end{itemize}
\begin{itemize}
	\item In most papers, to increase the impact of IRSs on system performance, direct link between the BS and the intended receivers is ignored. But, since both the base station and reflecting surfaces are usually located at high altitudes, the probability that both of them have direct links is high. Thus, in this paper, we suggest a more practical setup in which direct channels between BS and user and also BS and eavesdropper are also considered to exist.
\end{itemize}
\begin{itemize}
	\item We propose a new solution method for the beamforming and IRS phases optimization in the multi-IRSs networks. In this method, we divide our non-convex optimization problem into two sub-problems. After that, by definition of proper variables, Semi-Definite Relaxation (SDR) and Charness-Cooper Transformation (CCT), we propose an algorithm to achieve a sub-optimal solution. Then, we evaluate the complexity and convergence of this algorithm.  
\end{itemize}
%\begin{itemize}
%	\item Finally, Simulation results show the effectiveness of our proposed system in the network secrecy rate improvement.  
%\end{itemize}

The rest of the paper is organized as follows. In Section II, the system and channel model, as well as problem formulation, are introduced. In Section III, the proposed algorithm for solving the optimization problem is proposed and in Section IV, we present simulation results. Finally, Section V demonstrates the paper's conclusions.
%%%%%%%%%%%%%%%%%%%%%%%%%%

%%%%%%%%%%%%%%%%%%%%%%%%%%

\section{System Model, Channel Model and Problem Formulation} 
\label{sec:background}

\subsection{System Model}
As shown in Fig.1, we consider the downlink of a mmWave network where a BS serves a user in the presence of an eavesdropper. 
The BS is equipped with $M$ antennas and the user and eavesdropper each has one antenna. 
Also, each $L$ IRS has $N$ reflecting elements that assist in secure communication from the BS to the user. The reflecting elements of the IRSs change the phase of the received signal and then forward it towards the user. \textcolor{black}{Although, our main focus is on the single-user network, it can be further extended to a multi-user scenario. It should be noted that this single-user model can be applied to practical wireless systems in which the users are separated in time or frequency domains and at each resource block (assigned a single user) an eavesdropper tries to secretly listen to the data. }

 The channel coefficients $\Fb_l\in\mathbb{C}^{N \times M}$ , $\hb_d\in\mathbb{C}^{M \times 1}$, $\gb_d\in\mathbb{C}^{M \times 1}$, $\hb_{rl}\in\mathbb{C}^{N \times 1}$, and $\gb_{rl}\in\mathbb{C}^{N \times 1}$ are addressed the channel between the bases station and the $l$th IRS, the bases station and the user, the bases station and the eavesdropper, the $l$th IRS and the user, and the $l$th IRS and eavesdropper, respectively. We consider the perfectly estimated CSI at the BS. 
   The transmitted signal from the base station is denoted by $s$ that is with zero mean
   and unit variance. The beamforming vector at the BS is represented by  $\wb$ that satisfies $||\wb||^2 \le P_{BS}$, where $P_{BS}$ is the maximum allowable transmit power of the bases station.
	\begin{figure}[t]
		\centering
		\includegraphics[width=0.7\textwidth]{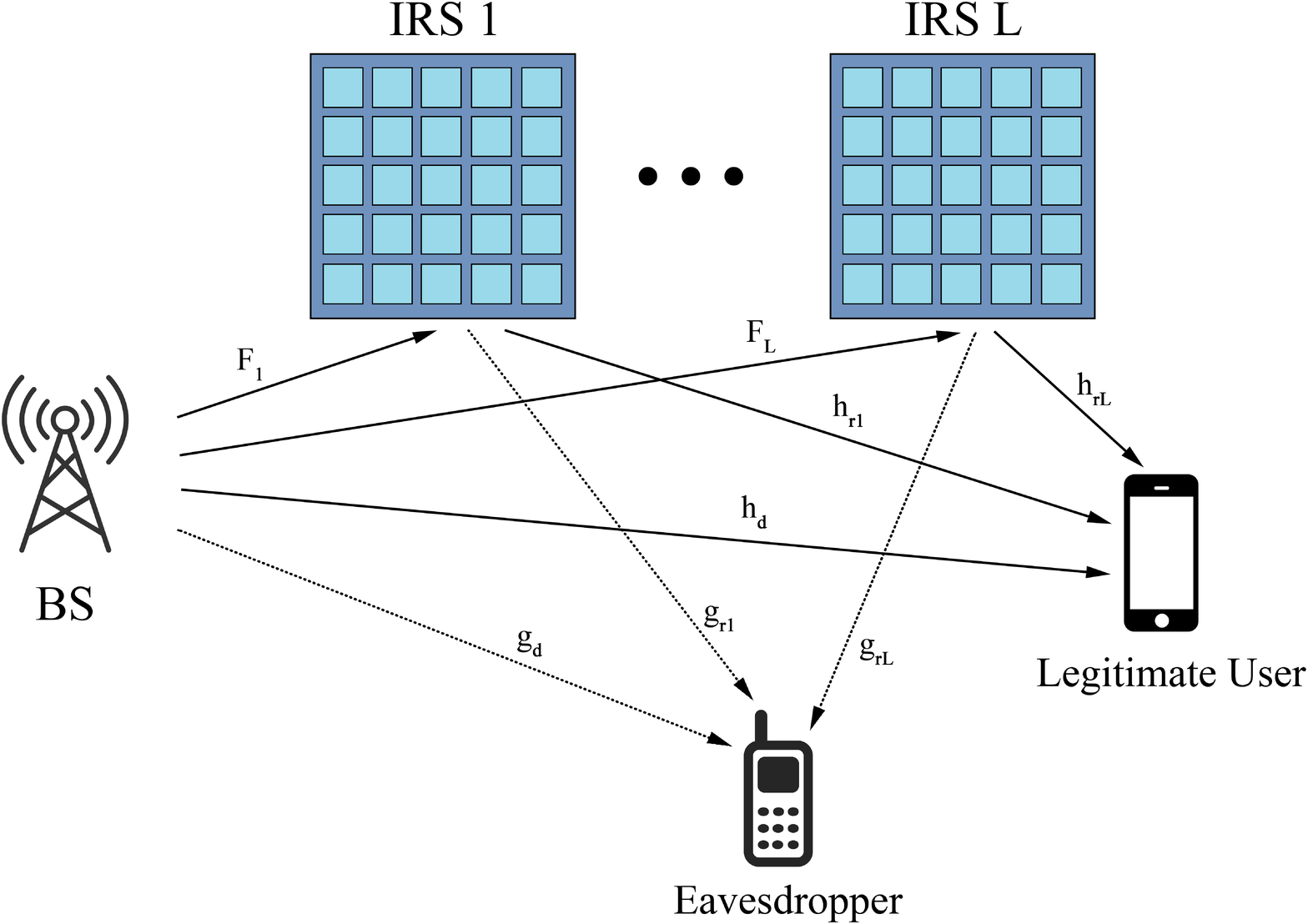}
		\caption{A multi IRS-aided mmWave communication network consisting a single legitimate user and an eavesdropper.}
		\label{fig_1}
	\end{figure}
	The received signal at the legitimate user can be written as
\begin{align}
y_r= \left( \sum_{l=1}^L \hb_{rl}^H \Thetab_l \Fb_l+\hb_d^H \right)  \wb s +n_r,
\end{align}
where $n_r$ with zero mean and variance $\sigma_r^2$ is Additive White Gaussian Noise (AWGN) at the user. \\ $\Thetab_l= \text{diag}  \{\:e^{j\theta_{l,1}}\:,\: e^{j\theta_{l,2}}, ...\: ,\: e^{j\theta_{l,N}}\: \} =\text{diag} \{\alpha_{l,1},\alpha_{l,2}, ... ,\alpha_{l,N}\} = \text{diag} (\boldsymbol{\alpha}_l)$ is  a diagonal matrix that denotes the effective phase shifts at the IRS elements and  $\theta_{l,n}\in[0,2\pi)$ is the phase shift of the $n$th reflecting element of the $l$th IRS. 
Similar to (1), the received signal at the eavesdropper is also obtained as
\begin{align}
y_e =\left( \sum_{l=1}^L \gb_{rl}^H \Thetab_l \Fb_l+\gb_d^H \right)  \wb s +n_e,
\end{align}
where $n_e$ is additive white Gaussian noise at the eavesdropper which has zero mean and its variance is $\sigma_e^2$.
\textcolor{black}{ In the system model that we consider in this paper, the effect of signal reflection between the IRSs is not considered. The reason is that this system is in the	mmWave frequency bands and the path loss is very severe in these frequencies. Hence, multiple reflections	by the IRS can be easily ignored.} 

Using (1) and (2), the achievable rates of the user and the eavesdropper are given by  
\begin{align}
R_r=\log_2(1+\frac{|\left(\sum_{l=1}^L \hb_{rl}^H \Thetab_l \Fb_l+\hb_d^H \right) \wb|^2}  {\sigma_r^2}),
\end{align}
and
\begin{align}
R_e=\log_2(1+\frac{|\left(\sum_{l=1}^L \gb_{rl}^H \Thetab_l \Fb_l+\gb_d^H\right)\wb|^2} {\sigma_e^2}),
\end{align}
respectively. Then,  the secrecy rate is obtained by $R_s = [R_r-R_e]^+$, where $[z]^+=$max$(z,0)$. Thus, the secrecy rate can be obtained as 
\begin{align}
R_s = [\log_2(1+\frac{|\left(\sum_{l=1}^L \hb_{rl}^H \Thetab_l \Fb_l+\hb_d^H\right)\wb|^2} {\sigma_r^2}) \nonumber\\
- \log_2(1+\frac{|\left(\sum_{l=1}^L \gb_{rl}^H \Thetab_l \Fb_l+\gb_d^H\right)\wb|^2} {\sigma_e^2})]^+.
\end{align}

\subsection{MmWave Channel Model}
The BS-to-user channel $\hb_d$ is modeled using a geometric channel model of the mmWave communications \cite{hs12} which can be expressed as 
\begin{align}
\hb_d = \frac{\sqrt M}{K} \sum_{k=1}^{K} \rho_{k,b}^u G_b  \ab_b(\phi_{k,b}^u),
\end{align}
where $K$ is total number of paths, $\rho_{k,b}^u \sim \mathcal{CN}(0,10^{-0.1PL(d)})$ is the complex gain of the $k$th path between BS and user, $G_b$ is the BS antenna gain. $\ab_b(\phi_{k,b}^u)$ is the normalized BS array response vector at an azimuth Angle of Departure AoD $\phi_{k,b}^u \in [0,2\pi]$.  The path loss can be obtained as \cite{hs13}
\begin{align}
PL(d) \: [dB] = \mu + 10 \kappa \: \log_{10} (d) + \xi,
\end{align}
in which $d$, $\mu$ and $\kappa$ denote the distance between transmitter and receiver, constant path loss term and path loss exponent, respectively. Also $ \xi \sim N(0,\sigma_{\xi}^2) $ where $\sigma_{\xi}^2$ is shadowing variance. 

Similar to (6), the BS-to-eavesdropper channel can be modeled by
\begin{align}
\gb_d = \frac{\sqrt M}{K} \sum_{k=1}^{K} \rho_{k,b}^e G_b \ab_b(\phi_{k,b}^e),
\end{align}
where  $\rho_{k,b}^e \sim \mathcal{CN}(0,10^{-0.1PL(d)})$ is the complex gain of the $k$th path between BS and eavesdropper. $\ab_b(\phi_{k,b}^e)$ is the normalized BS array response vector at an azimuth AoD $\phi_{k,b}^e \in [0,2\pi]$. 
%To determine both (6) and (8), the values of parameters for $PL(d)$ (7) are selected as $\mu = 72$, $\kappa = 2.92$ and $\sigma_{\xi} = 8.7$ dB \cite{hs13}. 

With the assumption that the BS and IRSs are located at a higher altitude, the channels between them are assumed to be LOS dominant. Thus the BS-to-$l$th IRS channel can be modeled as a rank-one matrix \cite{hs6}
\begin{align}
\Fb_l = \sqrt{MN} \rho_b^l G_b \ab_l(\phi_l^b,\psi_l^b) \ab_b^H(\phi_b^l), \: l \in \{1,2,.
..,L\},
\end{align}
where $\rho_b^l\sim \mathcal{CN}(0,10^{-0.1PL(d)})$ denotes the complex gain of the BS-to-$l$th IRS channel that $PL(d)$ is derived according to (7).
% with $\mu = 61.4$, $\kappa = 2$ and $\sigma_{\xi} = 5.8$ dB.  $\psi_l^b\in [0,\pi]$ and
 $\phi_l^b\in [0,\pi]$ is the elevation and azimuth Angles of Arrival (AoA) for the $l$th IRS and $\phi_b^l\in [0,2\pi]$ is the azimuth AoD for the BS. $\ab_l(\phi_l^b,\psi_l^b)$ is the normalized $l$th IRS array response vector, that is denoted by 
\begin{align}
\ab_l(\phi_l^b,\psi_l^b) = \ab_l^{az}(\psi_l^b) \otimes \ab_l^{el}(\phi_l^b),
\end{align}
where $\ab_l^{az}(\psi_l^b) $ and $\ab_l^{el}(\phi_l^b)$ are the horizontal and vertical array response vector of the $l$th IRS, respectively.
Also $\ab_b(\phi_b^l)$ is the normalized BS array response vector.

Similar to (6) and (8), the $l$th IRS to user and eavesdropper channels can be respectively given by 
\begin{align}
\hb_{rl} = \frac{\sqrt N}{K} \sum_{k=1}^{K} \rho_{k,l}^u G_l \ab_l(\phi_{k,l}^u,\psi_{k,l}^u),
\end{align}
and 
\begin{align}
\gb_{rl} = \frac{\sqrt N}{K} \sum_{k=1}^{K} \rho_{k,l}^e G_l \ab_l(\phi_{k,l}^e,\psi_{k,l}^e),
\end{align}
where $\rho_{k,l}^u \sim \mathcal{CN}(0,10^{-0.1PL(d)})$ and $\rho_{k,l}^e \sim \mathcal{CN}(0,10^{-0.1PL(d)})$ are the complex gain of the $k$th path of $l$th IRS-user channel and $l$th IRS-eavesdropper channel, respectively.
% that $PL(d)$ is derived according to (7) with $\mu = 72$, $\kappa = 2.92$ and $\sigma_{\xi} = 8.7$. 
   $\phi_{k,l}^u\in [0,\pi]$ and $\psi_{k,l}^u\in [0,\pi]$ are the azimuth and elevation (AoD) for the $l$th IRS, respectively. Also $\phi_{k,l}^e\in [0,\pi]$ and $\psi_{k,l}^e\in [0,\pi]$ are the azimuth and elevation AoD for the $l$th IRS.
 $\ab_l(\phi_{k,l}^u,\psi_{k,l}^u)$ and $\ab_l(\phi_{k,l}^e,\psi_{k,l}^e)$ are the normalized array response vectors of the $l$th IRS associated with the IRS-user and the IRS-eavesdropper paths, respectively. Definition of $\ab_l(\phi_{k,l}^u,\psi_{k,l}^u)$ and $\ab_l(\phi_{k,l}^e,\psi_{k,l}^e)$ are similar to (10) based on Kronecker product of the $l$th IRS horizontal and vertical array response vectors.
%%%%%%%%%%%%%%%%%%%%%%%%%
%%%%%%%%%%%%%%%%%%%%%%%%%
%%%%%%%%%%%%%%%%%%%%%%%%%
\subsection{Problem Formulation}
Based on the aforementioned discussion, this paper focuses on maximizing the secrecy rate by jointly optimizing the transmit beamforming vector $\wb$ at the BS and the phase vectors $\{\boldsymbol{\alpha}_1,...,\boldsymbol{\alpha}_L\}$ at the IRSs, under the BS's maximum transmission power and unit modulus of the diagonal elements of  $\{\Thetab_1,....\Thetab_L\}$ constraints. This problem can be expressed as 
\begin{subequations}
\begin{align}
&P: \: \: \max_{\wb,\boldsymbol{\alpha}_1,\boldsymbol{\alpha}_2,...,\boldsymbol{\alpha}_L} \:\: R_s \\
&\quad\quad\quad \text{s.t.} \:\:\: ||\wb||^2\le P_{BS}\\
&\quad\quad\quad |\alpha_{l,n}| = 1, \: \:   \:  l \in \{1,...,L\}, \: n \in \{1,...,N\} . 
\end{align}
\end{subequations}
Because of the coupled variables $\wb$ and $\{\boldsymbol{\alpha}_1,...,\boldsymbol{\alpha}_L\}$ in the objective and constraint functions, this problem $(P)$ is a non-convex problem.
%%%%%%%%%%%%%%%%%%
%%%%%%%%%%%%%%%%%%

\section{Proposed Algorithm And Beamforming Design}  
By decomposing  the optimization variables in (13) into two individual subsets of $\wb$ and $\{\boldsymbol{\alpha}_1,...,\boldsymbol{\alpha}_L\}$, it is found that (13b) and (13c) represents two disjoint sets on $\wb$ and  $\{\boldsymbol{\alpha}_1,...,\boldsymbol{\alpha}_L\}$, respectively. Thus, the problem can be solved alternatively by maximization of $\wb$ and $\{\boldsymbol{\alpha}_1,...,\boldsymbol{\alpha}_L\}$ in an iterative manner through two disjoint sub-problems. At each iteration, first for a given value of $\{\boldsymbol{\alpha}_1,...,\boldsymbol{\alpha}_L\}$, a solution for optimized $\wb$ is derived. Then, the second sub-problem is solved for the optimized $\wb$ and this procedure is iterated until convergence. Both sub-problems are solved using the SDR technique and CCT.
 \textcolor{black}{To extend the following proposed solution algorithm to the multi-user case, the most challenging part is to handle the interference among users. However, there are some ideas to overcome this challenge such as the transformation of objective to constraints and using successive convex approximation method \cite{SCA}.}

\subsection{Sub-Problem 1}
First, we assume that the parameters  $\{\boldsymbol{\alpha}_1,...,\boldsymbol{\alpha}_L\}$ are fixed and derive the optimal value of $\wb$. To this end, the objective function in problem $(P)$ is reformulated, and then considering that the log function is monotonically increasing, it is removed. Therefore,  sub-problem 1 will be
\begin{subequations}
	\begin{align}
		& \: \: \max_{\wb} \:\: (\frac {1+\frac{1}{\sigma_r^2} |\left(\sum_{l=1}^L \hb_{rl}^H \Thetab_l \Fb_l+\hb_d^H \right) \wb|^2 } {1+\frac{1}{\sigma_e^2} |\left(\sum_{l=1}^L \gb_{rl}^H \Thetab_l \Fb_l+\gb_d^H\right)\wb|^2} )\\
		&\quad\quad\quad \text{s.t.} \:\:\: ||\wb||^2\le P_{BS}. 
	\end{align}
\end{subequations}
By defining $\hb_{user}=\left(\sum_{l=1}^L \hb_{rl}^H \Thetab_l \Fb_l+\hb_d^H \right)\in\mathbb{C}^{1 \times M}$, $\gb_{eve}=\left(\sum_{l=1}^L \gb_{rl}^H \Thetab_l \Fb_l+\gb_d^H\right)\in\mathbb{C}^{1 \times M}$ and $\Wb = \wb \wb^H\in\mathbb{C}^{M \times M}$, we have 
\begin{align}
|\left(\sum_{l=1}^L \hb_{rl}^H \Thetab_l \Fb_l+\hb_d^H \right) \wb|^2 = \hb_{user} \Wb \hb_{user}^H
\end{align}
and 
\begin{align}
|\left(\sum_{l=1}^L \gb_{rl}^H \Thetab_l \Fb_l+\gb_d^H\right)\wb|^2 = \gb_{eve} \Wb \gb_{eve}^H.
\end{align}
We can rewrite (14) into a linear fractional problem as 
\begin{subequations}
	\begin{align}
	& \: P.1 : \: \max_{\Wb} \:\: (\frac {1+\frac{1}{\sigma_r^2}  \hb_{user} \Wb \hb_{user}^H } {1+\frac{1}{\sigma_e^2} \gb_{eve} \Wb \gb_{eve}^H} )\\
	&\quad\quad\quad \text{s.t.} \:\:\: \text{Tr}(\Wb)\le P_{BS}\\
	&\quad\quad\quad \Wb \succeq 0 \\
		&\quad\quad\quad \text{Rank}(\Wb)=1.
	\end{align}
\end{subequations}
\textcolor{black}{In this problem, we first use SDR to relax the rank-one constraint (17d). Then by employing the CCT method, $\Wb$ is obtained, and finally,  to satisfy the (17d), a rank one approximation of $\Wb$ is derived using the Gaussian randomization method.} \\ 
By dropping (17d), using the CCT
and defining $\gamma = \frac{1} {1+\frac{1}{\sigma_e^2}  (\gb_{eve} \Wb \gb_{eve} ^H)}$ and $ \Tb = \gamma \Wb $, (17) is transformed into the following non-fractional problem 
	\begin{subequations}
	\begin{align}
	\max_{\Tb, \gamma}  \gamma + \frac{1}{\sigma_r^2}  ( \hb_{user} \Tb  \hb_{user} ^H ) \quad\quad\\ 
	\text{s.t.} \quad \gamma+\frac{1}{\sigma_e^2}(\gb_{eve}\Tb\gb_{eve}^H) = 1 \\
	 \text{Tr}(\Tb) \le \gamma P_{BS} \\
	\Tb \succeq 0, \gamma \ge 0.
	\end{align}
\end{subequations}
The problem in (18) is a Semi-Definite Programming (SDP) problem and therefore it is convex and CVX can be used to solve such a problem. After solving (18), $\Wb$ is obtained as $T/\gamma$. Then, to satisfy the constraint Rank$(\Wb)=1$, according to \cite{gaussian} we use the standard Gaussian randomization method to achieve an approximation for $\wb$. 
%%%%%%%%%%%%%%%%%%%
%%%%%%%%%%%%%%%%%%%
%%%%%%%%%%%%%%%%%%
\subsection{Sub-Problem 2}
	In the next step, for a given $\wb$, the problem $(P)$ is transformed to sub-problem 2 that can be formulated as follows \cite{hs9}
   	\begin{align}
   \max_{\boldsymbol{\alpha}_1,\boldsymbol{\alpha}_2,...,\boldsymbol{\alpha}_L} \frac{(1+\frac{1}{\sigma_r^2}|\left(\sum_{l=1}^L \hb_{rl}^H \Thetab_l \Fb_l+\hb_d^H\right)\wb|^2)} {(1+\frac{1}{\sigma_e^2}|\left(\sum_{l=1}^L \gb_{rl}^H \Thetab_l \Fb_l+\gb_d^H\right)\wb|^2)} \quad\quad\quad\quad \nonumber\\
   \text{s.t.} \: \: \: |\alpha_{l,n}| = 1, \: \:   \:  l \in \{1,2,...,L\}, \: n \in \{1,2,...,N\} . 
   \end{align}

   \textcolor{black}{Then, using the equality $\cb \Thetab_l  = \boldsymbol{\alpha_l}^T$ diag$(\cb)$, where $\Thetab$ is analog beamforming matrix and $\boldsymbol\alpha_l$ is the phase vector of the $l$-th IRS, and also defining $\cb=\hb_{rl}^H$ and $\cb=\gb_{rl}^H$, we can write  
    	\begin{align}
    \sum_{l=1}^L \left(\hb_{rl}^H \Thetab_l \Fb_l + \frac{1}{L}\hb_d^H\right)\wb=\sum_{l=1}^L\left(\boldsymbol{\alpha_l}^T \text{diag} (\hb_{rl}^H)\Fb_l+\frac{1}{L}\hb_d^H\right)\wb
    	\end{align} 
	\begin{align}
	\sum_{l=1}^L \left(\gb_{rl}^H \Thetab_l \Fb_l+\frac{1}{L}\gb_d^H\right)\wb=\sum_{l=1}^L \left(\boldsymbol{\alpha_l}^T \text{diag} (\gb_{rl}^H)\Fb_l+\frac{1}{L}\gb_d^H\right)\wb.
	\end{align}
}
\textcolor{black}{Then, (20) and (21) can be further simplified as follows
	\begin{align}
	\sum_{l=1}^L \left(\hb_{rl}^H \Thetab_l \Fb_l + \frac{1}{L}\hb_d^H\right)\wb = \xb_l^H \hb^\prime_{user,l}
	\end{align}
	\begin{align}
	\sum_{l=1}^L \left(\gb_{rl}^H \Thetab_l \Fb_l+\frac{1}{L}\gb_d^H\right)\wb =\xb_l^H \gb^\prime_{eve,l}
	\end{align}
}
where	 $\xb_l=[\boldsymbol{\alpha_l}^T,\frac{1}{L}]^H\in\mathbb{C}^{(N+1) \times 1}$, and  
\begin{subequations}
	\begin{align}
	\hb^\prime_{user,l} = 
	\begin{pmatrix}
\text{diag}(\hb_{rl}^H) \Fb_l \wb  \\
	\hb_d^H \wb 
	\end{pmatrix} \in\mathbb{C}^{(N+1) \times 1}\\
	\gb^\prime_{eve,l} = 
	\begin{pmatrix}
\text{diag}(\gb_{rl}^H) \Fb_l \wb  \\
	\gb_d^H \wb 
	\end{pmatrix} \in\mathbb{C}^{(N+1) \times 1}.
	\end{align}
\end{subequations} 
\textcolor{black}{By substitution of (22) and (23) into objective function of (19), sub-problem 2 can be reformulated as} 
    \begin{subequations}
	\begin{align}
	\max_{\xb_1,\xb_2,...,\xb_L} \: \: \: \frac{(1+\frac{1}{\sigma_r^2}\sum_{l=1}^L |\xb_l^H \hb^\prime_{user,l}|^2)} {(1+\frac{1}{\sigma_e^2}\sum_{l=1}^L |\xb_l^H \gb^\prime_{eve,l}|^2)}\\
	\quad\quad\quad\quad \text{s.t.} \: \: \: \xb_l^H \Sb_i \xb_l = 1,\forall i,l \quad\quad\quad\quad\quad\quad\quad\quad
	\end{align}
\end{subequations}
\textcolor{black}{where $\Sb_i$ is defined to satisfy the unit modules constraint of the IRS element gain and its $(k,j)$th element is denoted by} 
    	\begin{align}
    	[\Sb_i]_{k,j} = \left\{
    	\begin{array}{lr}
    	1 &i = k = j , \: i \in \{1,...,L-1\} \\
    	L^2 & i = k = j = L \quad\quad\quad\quad\quad\quad \\
    	0 &\text{otherwise} \quad\quad\quad\quad\quad \quad\quad\quad
    	\end{array}\right.
    	\end{align}
The objective function (25a) is a quadratically fractional function and the constraint (25b) is quadratic equality and non-convex and hence, the problem (25) is not a convex problem. 
%and an approximate solution is presented using SDR technique as follows. %Problem (19) can be handled by the SDR.% 

By defining $\Xb_l = \xb_l \xb_l^H $, the numerator of (25a) is expressed
as
\begin{IEEEeqnarray}{lCr}
1+\frac{1}{\sigma_r^2}\sum_{l=1}^L |\xb_l^H \hb^\prime_{user,l}|^2 = 1+\frac{1}{\sigma_r^2}\sum_{l=1}^L \hb_{user,l}^{\prime H} \xb_l \xb_l^H \hb^\prime_{user,l}= \nonumber\\  1+\frac{1}{\sigma_r^2}\sum_{l=1}^L \hb_{user,l}^{\prime H} \Xb_l \hb^\prime_{user,l} \quad\quad\quad\quad\quad\quad\quad 
\end{IEEEeqnarray}
and denominator of (25b) is expressed as 
\begin{IEEEeqnarray}{lCr}
1+\frac{1}{\sigma_e^2}\sum_{l=1}^L |\xb_l^H \gb^\prime_{eve,l}|^2 = 1+\frac{1}{\sigma_e^2}\sum_{l=1}^L \gb_{eve,l}^{\prime H} \xb_l \xb_l^H \gb^\prime_{eve,l}= \nonumber \\  1+\frac{1}{\sigma_e^2}\sum_{l=1}^L \gb_{eve,l}^{\prime H} \Xb_l \gb^\prime_{eve,l}  
\end{IEEEeqnarray}
According to (28) and (29), we rewrite (25) into a linear fractional problem as %by dropping the constraint of $\text{Rank}(\Xb) = 1$ as follows%
   \begin{subequations}
\begin{align}
\max_{\Xb_1,...,\Xb_L } \: \: \frac{1+\frac{1}{\sigma_r^2} \sum_{l=1}^L(\hb_{user,l}^{\prime H} \Xb_l \hb^\prime_{user,l})} {1+\frac{1}{\sigma_e^2} \sum_{l=1}^L(\gb_{eve,l}^{\prime H} \Xb_l \gb^\prime_{eve,l})} \\ 
\text{s.t.} \: \: \: \text{tr}(\Sb_i \Xb_l) = 1 , \forall i,l,  \quad\quad\quad\quad\\
\text{Rank}(\Xb_l) = 1, \forall l, \quad\quad\quad\quad \:\: \: \\ 
\Xb_l \succeq  0.  \quad\quad\quad\quad\quad\quad\quad\quad
\end{align}
      \end{subequations} 
  \textcolor{black}{Similar to the first sub-problem, we drop rank one constraint, then use CCT to obtain $\Xb_l, \forall l$ and finally by employing the Gaussian randomization method, we find an approximated rank-one solution.}
  
In the following, by dropping (27c), using the CCT and defining $\lambda = 1 / (1+\frac{1}{\sigma_e^2} \sum_{l=1}^L(\gb_{eve,l}^{\prime H} \Xb_l \gb^\prime_{eve,l}))$ and $ \Tb_l = \lambda \Xb_l $, (30) is transformed into the following non-fractional problem   
  \begin{subequations}
  	\begin{align}
  	P.2 : \: \max_{\Tb_1,...,\Tb_L , \lambda }  \lambda+\frac{1}{\sigma_r^2} \sum_{l=1}^L(\hb_{user,l}^{\prime H} \Tb_l \hb^\prime_{user,l}) \quad\quad\\ 
  	\text{s.t.} \quad \lambda+\frac{1}{\sigma_e^2} \sum_{l=1}^L(\gb_{eve,l}^{\prime H} \Tb_l \gb^\prime_{eve,l}) = 1  \\ 
  	\text{tr}(\Sb_i \Tb_l) =  \lambda , \: \forall i,l \quad\quad\quad\quad\quad\quad\\ 
  	\Tb_l \succeq 0, \forall l\quad\quad\quad\quad\quad\quad\quad\quad\quad\\
  	\lambda \ge 0  \quad\quad\quad\quad\quad\quad\quad\quad\quad
  	\end{align}
  \end{subequations} 
The problem in (31) is an SDP problem and therefore a convex optimization problem.
After solving (31), $\Xb_l$ is obtained as $\Tb_l/\lambda$. Similar to $\wb$, an approximate solution to $\xb_l$ can be obtained by standard Gaussian randomization. 

Our final proposed algorithm to find the solution of problem $P$ is presented in \textbf{Algorithm 1}.
   
   \begin{algorithm}[H]
   	\caption{Proposed alternating iterative algorithm for solving P.}
   	\begin{algorithmic}[1]
   		\renewcommand{\algorithmicensure}{\textbf{Output:}}
   		\renewcommand{\algorithmicrequire}{\textbf{Initialization:}}
   		\ENSURE  $ {\bf{w}},{\Thetab_1},{\Thetab_2},...,{\Thetab_L}, R_s. $			
   		\REQUIRE $i = 0,\wb^{(0)},{\Thetab_1^{(0)}}={\Thetab_2^{(0)}}=...={\Thetab_L^{(0)}},$
   		$\varepsilon  = {10^{ - 3}}.$
   		\STATE \textbf{Repeat}
   		\STATE Set $ i=i+1 $.
   		\STATE Using $ {\Thetab_1^{(i-1)}},{\Thetab_2^{(i-1)}},...,{\Thetab_L^{(i-1)}} $,  Solve $P.1$ and obtain ${\bf{T}^{(i)}}$ and ${\bf{\gamma}^{(i)}}$.
   		\STATE Set ${\bf{W}^{(i)}}={\bf{T}^{(i)}}/{\bf{\gamma}^{(i)}}$ and 
   		derive ${\bf{w}^{(i)}}$ by employing the Gaussian Randomization method. 
   		\STATE With given ${\bf{w}^{(i)}}$, Solve the $P.2$ and find $\Tb_1^{(i)},\Tb_2^{(i)},...,\Tb_L^{(i)}$ and ${\bf{\lambda}^{(i)}}$.
   		  \STATE Set $\Xb_l^{(i)}=\Tb_l^{(i)}/{\bf{\lambda}^{(i)}}$; $\forall l \in \{1,2,,,,L\}$ and derive ${\Thetab_1^{(i)}},{\Thetab_2^{(i)}},...,{\Thetab_L^{(i)}} $ by employing the Gaussian Randomization method.
   		\STATE \textbf{Until} $ \left| {\frac{{ R_s^{(i)} - R_s^{(i-1)}}}{{R_s^{(i-1)}}}} \right| \le {\varepsilon } $.
   	\end{algorithmic}
   \end{algorithm}

\subsection{Convergence of the Proposed Algorithm}
 In the following theorem, we discuss the convergence behavior of our proposed algorithm.
 \begin{theorem}
 	Our proposed algorithm converges to a finite value in a non-decreasing fashion.	
 \end{theorem}
 \begin{proof} 
 	Based on the following equation we find that the objective function has a non-decreasing behavior at the successive iterations.
 	\begin{align}
 		R_s(\bf{W}^{(i-1)},{\Thetab_0^{(i-1)}},...,{\Thetab_L^{(i-1)}}) \le R_s(\bf{W}^{(i)},{\Thetab_0^{(i-1)}},...,{\Thetab_L^{(i-1)}}) \le R_s(\bf{W}^{(i)},{\Thetab_0^{(i)}},...,{\Thetab_L^{(i)}})
 	\end{align}
 	The first inequality follows the fact that at the i-th iteration in the first sub-problem, for the given $\Thetab$ which is derived from the last iteration, $\bf{W}$ is optimized, thus the objective function improves. In the second inequality, this process is repeated for the given  $\bf{W}$ and the objective function grows by optimizing $\Thetab$. Due to the limited resource such as power, the number of antennas, and IRSs, the final value of the secrecy rate is upper bounded and the algorithm converges to a finite value in a non-decreasing fashion.	
 \end{proof}

\subsection{Complexity of the Proposed Algorithm}
Algorithm 1 requires solving two convex problems (18) and (30), at each iteration. The complexity of steps (3) and (4), which is related to the convex problem (18), is
$\mathcal{O}(M^{3.5})$ \cite{polik2010interior}. The complexity of steps (5) and (6) that express the convex problem (30), is $\mathcal{O}(L(N+1)^{3.5})$. As a result, the total complexity of  Algorithm 1  is $\mathcal{O}(I_{itr}(M^{3.5}+L(N+1)^{3.5}))$, where $I_{itr}$ represents the iterations number until the convergence criterion is met. \textcolor{black}{This shows that by employing multiple IRSs, computational complexity increases linearly.} 
   
\section{Simulation Results And Discussions}
In this section, we present the simulation results of evaluating our proposed algorithm and indicate the advantages of using  multiple IRSs in improving the secrecy rate of the mmWave networks and also the effect of optimal designing of these surfaces. 
In our simulations, we assume that a BS with $M$ antennas is located at the center of the polar coordinates. In addition, $L$  IRSs each with $N$ reflecting elements are installed around the BS at the fixed locations to assist in signal transmissions. We assume that the IRSs are located on a circle centered at the BS but with different angle. The $ L $ ($ L \in \{1,3,5\} $) IRSs are placed at $(25\cos(\frac{\pi} {4} + i \frac{\pi}{12}),25\sin(\frac{\pi} {4} + i \frac{\pi}{12}))$ for all $i \in \{-\frac{(L-1)}{2},-\frac{(L-3)}{2},...,0,...,\frac{(L-3)}{2},\frac{(L-1)}{2}\}$. Also the user and eavesdropper are located at $(20 , \beta)$ and $(18 , \beta)$, respectively, where $\beta \sim U[0,\pi/2] $. The noise variances are set as $\sigma_r^2=\sigma_r^2=-95$dBm. $PL(d)[dB]$ is calculated as in (6), (8), (11) and (12), based on $\mu=72$, $\kappa=2.92$ and $\sigma_{\zeta}=8.7$ and for (9) based on $\mu=61.4$, $\kappa=2$ and $\sigma_{\zeta}=5.8$ \cite{hs13}.
 \clearpage
\begin{figure}[h]
	\centering
	\includegraphics[width=12cm, height=9cm]{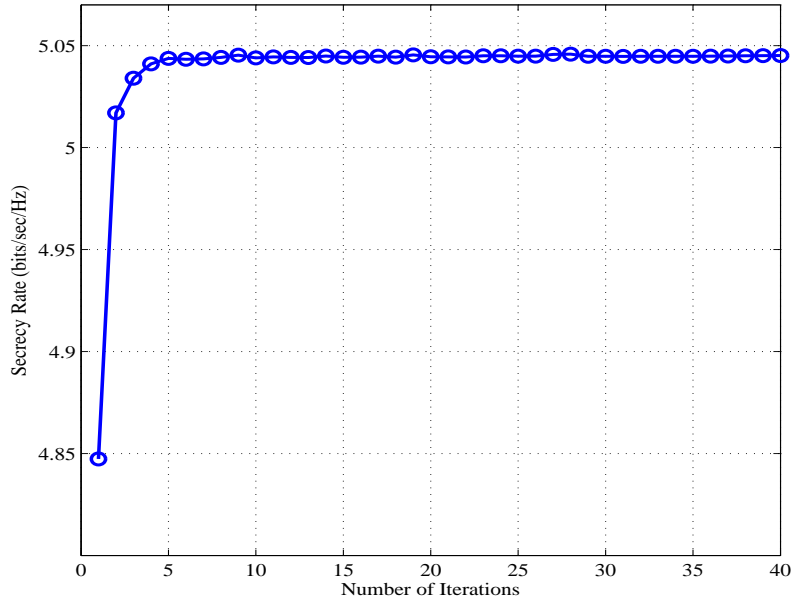}
	\caption{ Convergence of our proposed algorithm}
	\label{fig_1}
\end{figure}
Maximum Ratio Transmission (MRT) beamforming and No-IRS system are selected as benchmarks for our system performance evaluation. In MRT, conjugate of the channel between the BS and legitimate user is considered as beamforming vector. After beamforming determination, the reflecting elements phase shifts are obtained using Algorithm 1. Both optimization 
	\begin{figure}[hb]
		\centering
		\includegraphics[width=12cm, height=9cm]{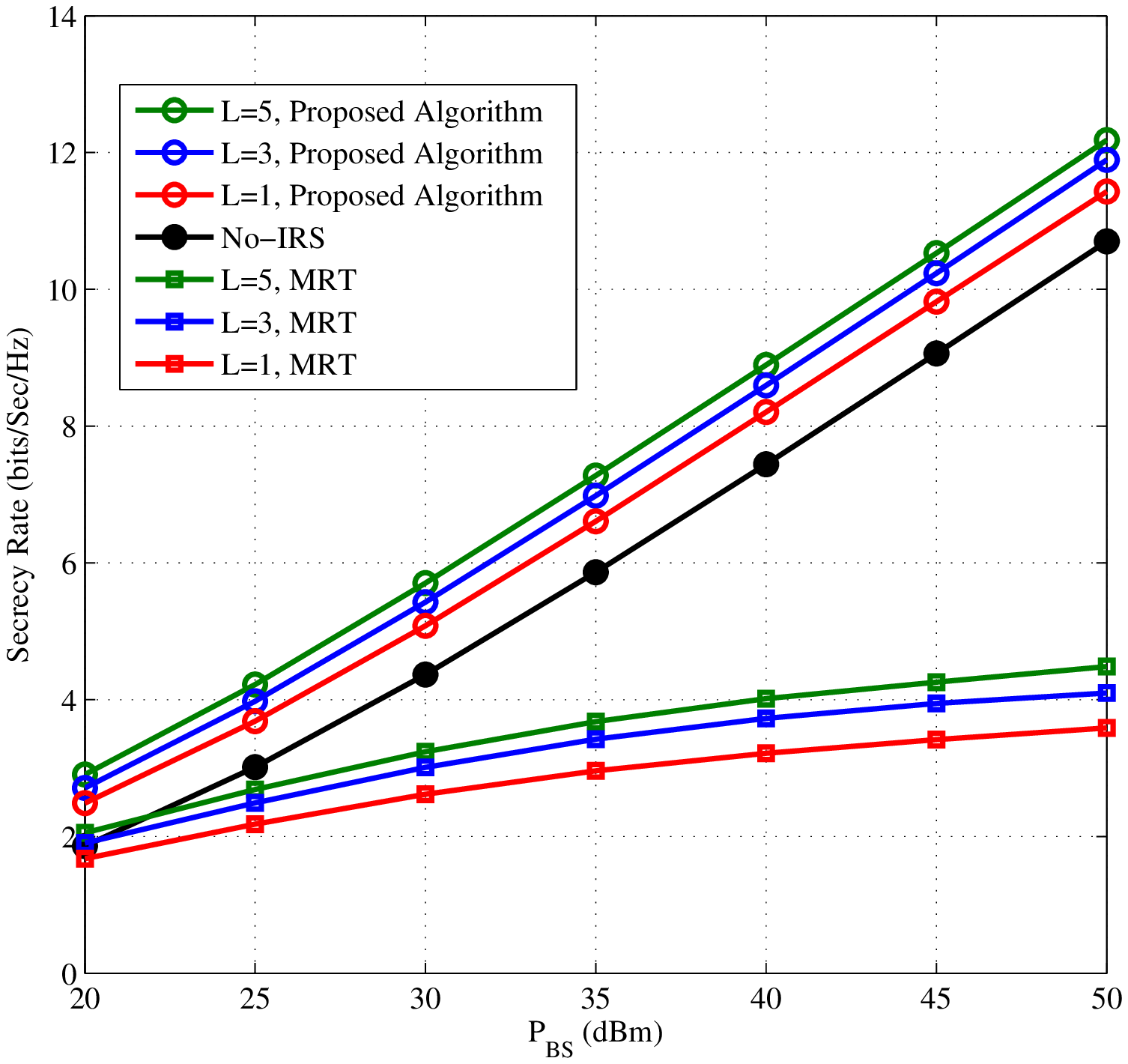}
		\caption{\textcolor{black}{Secrecy rate versus the BS's transmission power ($P_{BS}$) for different number of IRSs (with M=4, N=16)}}
		\label{fig_1}
	\end{figure}
 sub-problems are solved using the CVX optimization toolbox and all the results are averaged over 1000 iterations. 

In Fig. 2, the convergence behavior of our proposed algorithm (Algorithm 1) is evaluated. As we discussed in the previous section, the secrecy rate is a non-decreasing function with respect to the number of iterations, thus as it is shown in this figure, by increasing the number of iterations, the secrecy rate also increases and converges to its maximum value after about five iterations.
  \begin{figure}[t]
	\centering
	\includegraphics[width=12cm, height=10cm]{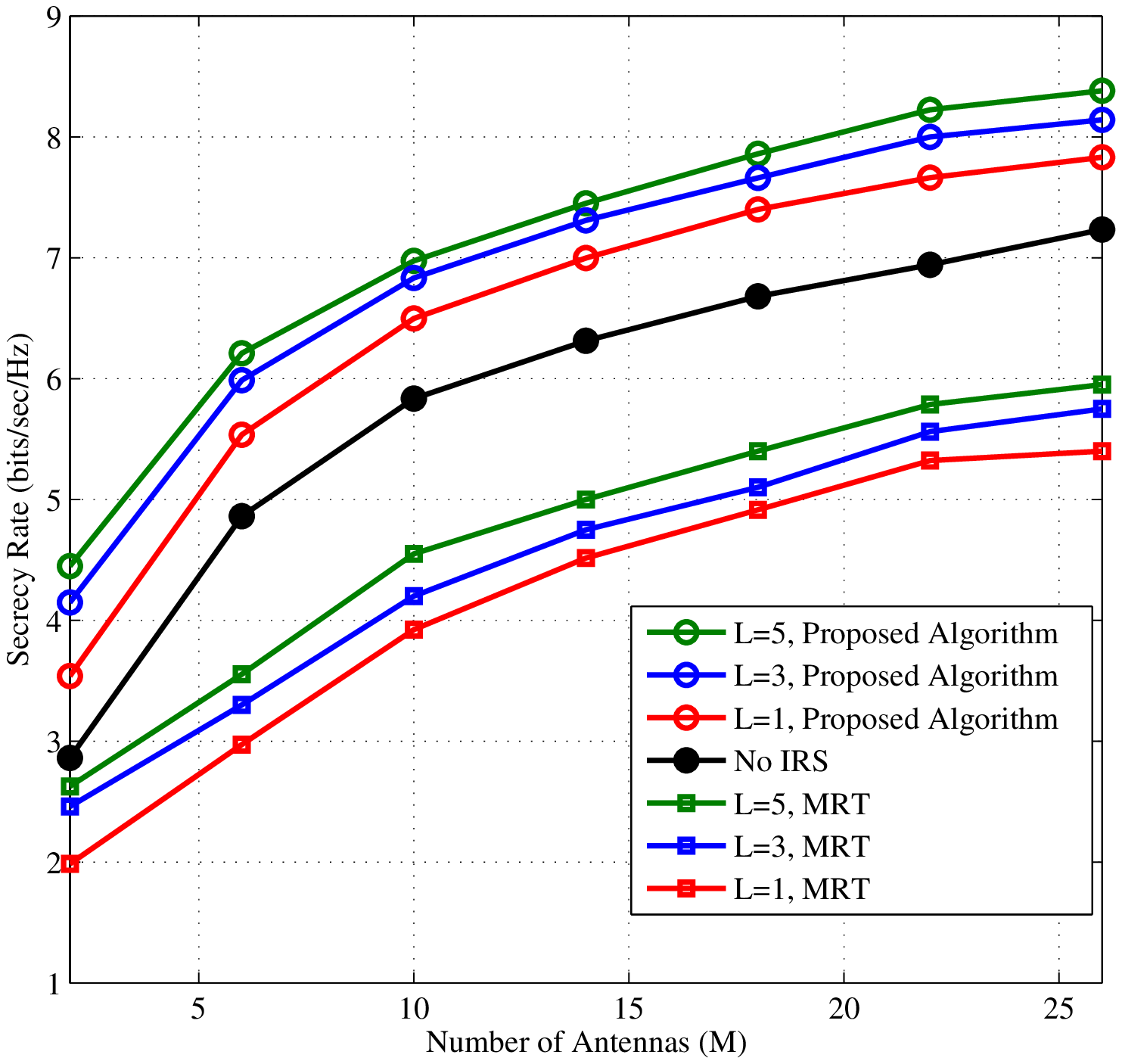}
	\caption{\textcolor{black}{Secrecy rate versus the number of BS's antennas for different number of IRSs (with $P_{BS}$ = 35dBm).}}
	\label{fig_1}
\end{figure}

Fig. 3 shows secrecy rate against BS's transmit power for three different numbers of IRSs $(L=1,3,5)$ \textcolor{black}{as well as no-IRS case}, where the number of antennas ($ M$) and reflecting elements in each IRS ($ N $) are set to 4 and 16, respectively. It can be seen that an increase in the number of IRSs results in secrecy rate improvement. This is because a large number of IRSs can lead to strengthening the signal at the user and better user's signal suppression at the eavesdropper. \textcolor{black}{Another interesting point is that the gap between curves decreases by adding more IRSs which means that change the system from a no-IRS case to the single IRS case has the best system performance improvement and by adding more IRS the value of this improvement is reduced. Since in this paper we tried to consider more practical concerns, we didn't ignore the effect of the channel between the BS and legitimate user and eavesdropper which may have the same probability of existence as the IRS channels. Furthermore, due to BS has multiple antennas and can perform active beamforming, the great significance of the BS beamforming with respect to IRS beamforming can be seen by comparing of the no-IRS curve and MRT related curves. Actually, when MRT is used, the objective to perform active beamforming isn't the secrecy rate, and also channel between BS and IRSs are not considered. As a result, the power received by the IRS reduced and they do not play important role in secrecy rate improvement. Therefore, the no-IRS case has better performance than the MRT case with IRS. In contrast to the aforementioned drawbacks of MRT, it has two main advantages. Firstly, this beamforming has lower computational complexity with respect to the optimal beamforming. Secondly, this method can separate IRSs and BS in solving optimization problems which leads to optimization problems that can be handled in a distributed manner.} 
 
 %The BER of MRT and proposed EGT schemes, as a function of SNR

 \begin{figure}[t]
 	\centering
 	\includegraphics[width=12cm, height=10cm]{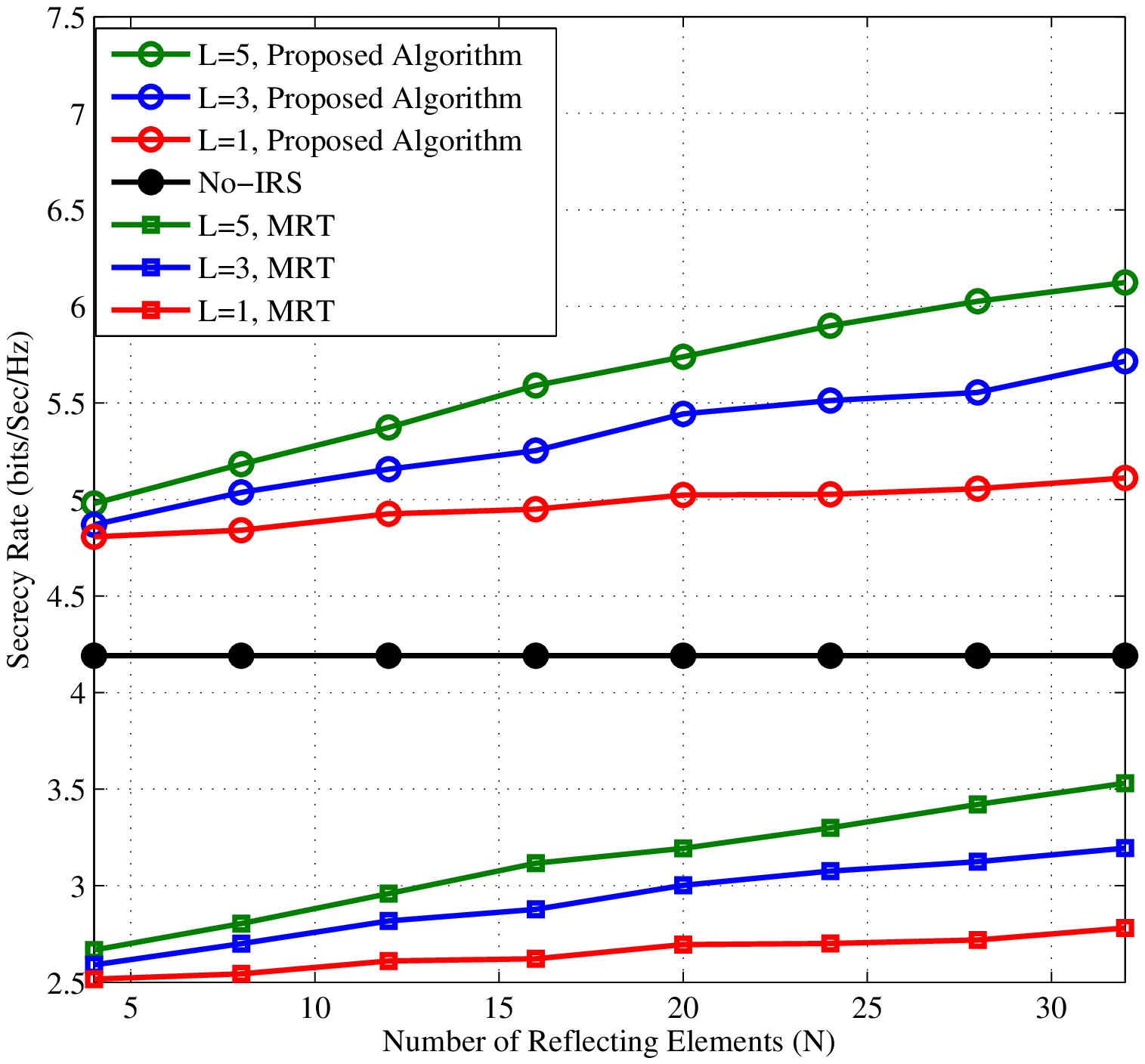}
 	\caption{\textcolor{black}{Secrecy rate versus the number of reflecting elements at each IRS for different number of IRSs.}}
 	\label{fig_1}
 \end{figure}

Fig. 4 illustrates the effect of the number of antennas at the BS on the network secrecy rate for different numbers of IRSs, \textcolor{black}{no-IRS case} and also the MRT method. It can be observed due to the more active beamforming gain of BS at a higher number of antennas, the secrecy rate improves by increasing the number of antennas. \textcolor{black}{Similar to the previous figure, the gap between curves by increasing the number of IRSs is reduced. Also due to the earlier mentioned reasons, the MRT method has the least secrecy rate compared with a system with optimal active beamforming.} 

Fig. 5, shows the secrecy rate versus the number of reflecting elements at the IRSs. 
As it is expected, with the higher number of IRSs elements, a better secrecy rate is attained. Also, add more reflecting elements introduces more effectiveness of the number of IRSs that is understandable from the gap intensification between the curves. \textcolor{black}{Since in no-IRS case there is no reflecting element, secrecy rate is constant by increasing the number of reflecting elements. Another interesting point that is at a higher number of reflecting elements, IRSs play a more important role in system performance and the gap between MRT and the no-IRS case is reduced and this gap among optimal beamforming with IRSs and no-IRS case increases.} 
\begin{figure}[t]
	\centering
	\includegraphics[width=12cm, height=10cm]{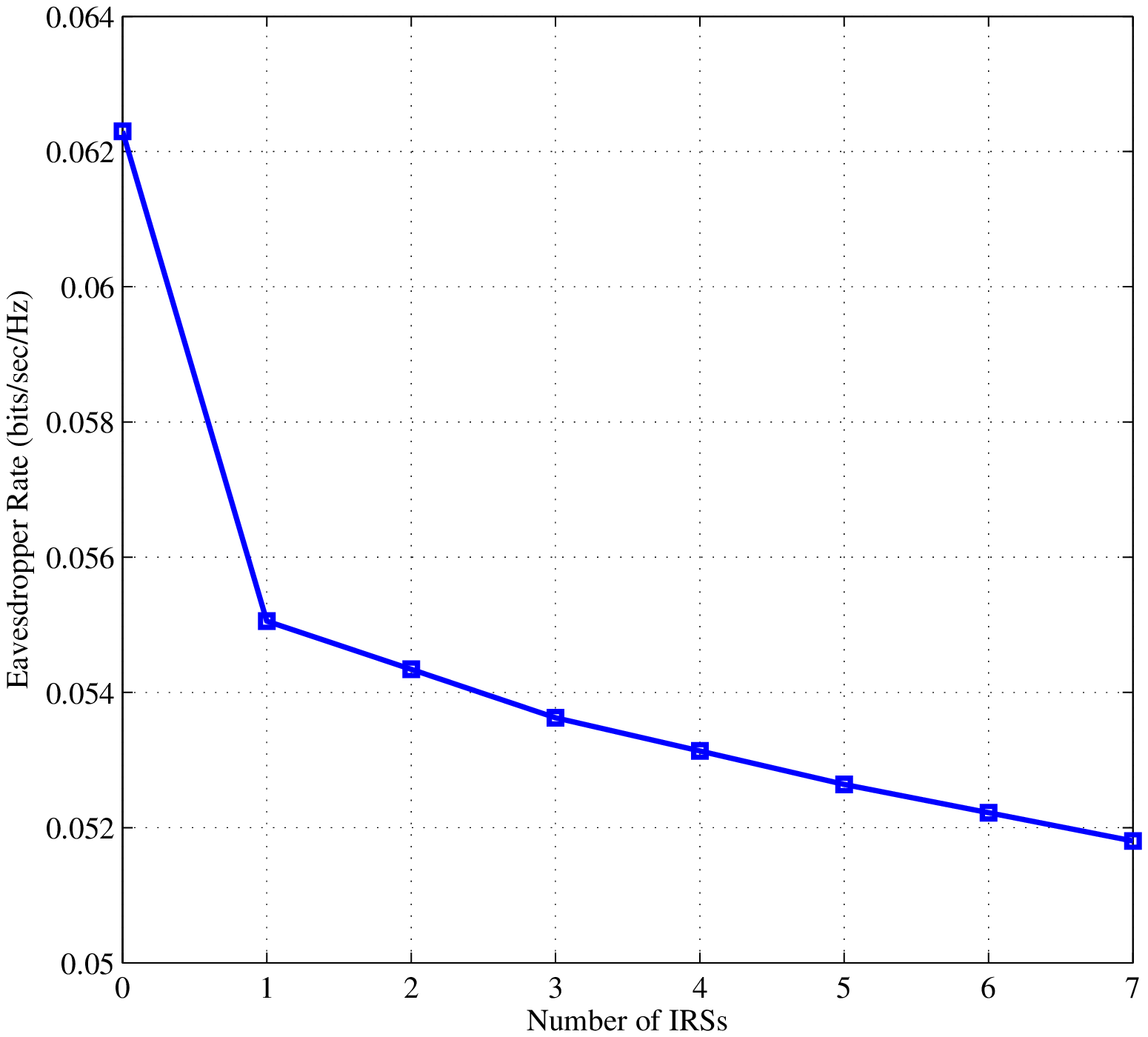}
	\caption{\textcolor{black}{Eavesdropper rate versus the number of IRSs.} }
	\label{fig_1}
\end{figure}	

\textcolor{black}{Fig.6 shows the effect of adding more IRSs on eavesdropper rate. As it is shown in this figure and by considering the previous results, we can find that increasing the number of IRSs not only increases the secrecy rate but also decreases the eavesdropper rate. This decrement is reduced at the higher number of IRSs.} 
%%%%%%%%%%%%%%%%
%%%%%%%%%%%%%%%%
%%%%%%%%%%%%%%%%
%%%%%%%%%%%
\section{CONCLUSIONS}
In this paper, we investigated a multi-IRS mmWave system, where the BS's beamforming vector and the IRSs reflecting elements' phases were jointly optimized to maximize the network secrecy rate. We solved the resulting non-convex optimization problem using a novel method based on alternating technique and SDR method. Simulation results showed that our proposed method outperforms the case that MRT is used as the beamforming vector at the BS. In addition, we demonstrated that adding more IRSs in the network or increasing the number of elements at each IRS can enhance the performance of the mmWave networks in terms of the secrecy rate.   

\bibliographystyle{IEEEtran}
\bibliography{IEEEabrv,library}

\end{document}